\documentclass[11pt]{article}

\usepackage[margin=1in]{geometry}

\usepackage[utf8]{inputenc}
\usepackage{amsmath}
\usepackage{amssymb}
\usepackage{amsthm}
\usepackage{mathtools}
\usepackage{bbm}
\usepackage{bm}
\usepackage{color}
\usepackage{graphicx}
\usepackage{array}
\usepackage{multirow}
\usepackage{enumitem}
\usepackage[ruled,vlined]{algorithm2e}
\usepackage{relsize}
\usepackage[usenames,dvipsnames]{xcolor}
\usepackage[linktocpage=true,
	pagebackref=true,colorlinks,
	linkcolor=BrickRed,citecolor=blue]
	{hyperref}
\usepackage{cleveref}
\usepackage[nottoc]{tocbibind}

\usepackage{thmtools,thm-restate}
\usepackage{nicefrac}

\DeclareMathOperator*{\E}{\mathbb{E}}

\DeclareMathOperator*{\Cov}{Cov}
\DeclareMathOperator{\given}{\:\vert\:}
\DeclareMathOperator*{\argmax}{argmax}
\DeclareMathOperator{\OPT}{OPT}
\DeclareMathOperator{\ALG}{ALG}

\newcommand{\calA}{\mathcal{A}}
\newcommand{\calB}{\mathcal{B}}
\newcommand{\calD}{\mathcal{D}}

\newcommand{\calI}{\mathcal{I}}
\newcommand{\calM}{\mathcal{M}}
\newcommand{\calP}{\mathcal{P}}

\newcommand{\IGNORE}[1]{}

\newcounter{note}[section]

\newcommand{\PreserveBackslash}[1]{\let\temp=\\#1\let\\=\temp}
\newcolumntype{C}[1]{>{\PreserveBackslash\centering}p{#1}}
\newcolumntype{R}[1]{>{\PreserveBackslash\raggedleft}p{#1}}
\newcolumntype{L}[1]{>{\PreserveBackslash\raggedright}p{#1}}

\newtheorem{theorem}{Theorem}[section]
\newtheorem{claim}[theorem]{Claim}
\Crefname{claim}{Claim}{Claims}
\newtheorem{proposition}[theorem]{Proposition}
\newtheorem{lemma}[theorem]{Lemma}

\theoremstyle{definition}

\newtheorem{definition}[theorem]{Definition}

\newtheorem{remark}[theorem]{Remark}

\makeatletter
\newcommand{\namelabel}[1]{%
  \phantomsection
  \renewcommand{\@currentlabel}{#1}
  \label{#1}
}
\makeatother

\usepackage{tikz}
\usetikzlibrary{decorations.markings}
\usetikzlibrary{arrows}
\tikzstyle{block}=[draw opacity=0.7,line width=1.4cm]
\tikzstyle{graphnode}=[circle, draw, fill=black!20, inner sep=0pt, minimum width=6pt]
\tikzstyle{point}=[circle, draw, fill=black!30, inner sep=0pt, minimum width=1pt]
\tikzstyle{input}=[rectangle, draw, fill=black!75,inner sep=3pt, inner ysep=3pt, minimum width=4pt]
\tikzstyle{unmatched}=[graphnode,fill=black!0]
\tikzstyle{shaded}=[graphnode,fill=black!20]
\tikzstyle{matched}=[graphnode,fill=black!100]  	
\tikzstyle{matching} = [ultra thick]
\tikzset{
    >=stealth',
    pil/.style={
           ->,
           thick,
           shorten <=2pt,
           shorten >=2pt,}
}
\tikzset{->-/.style={decoration={
  markings,
  mark=at position .5 with {\arrow{>}}},postaction={decorate}}}

\title{Submodular Dominance and Applications}
\author{Frederick V. Qiu\thanks{(fqiu@princeton.edu)  Department of Computer Science,  Princeton University.} \and Sahil Singla\thanks{(ssingla@gatech.edu)  School of Computer Science,  Georgia Tech.}}
\date{\today}

\usepackage{tocloft}

\begin{document}

\maketitle

\begin{abstract}
\medskip
In submodular optimization we often deal with the expected value of a submodular function $f$ on a distribution $\calD$ over sets of elements. In this work we study such submodular expectations for negatively dependent distributions. We introduce a natural notion of negative dependence, which we call \emph{Weak Negative Regression} (WNR), that generalizes both Negative Association and Negative Regression. We observe that WNR distributions satisfy \emph{Submodular Dominance}, whereby the expected value of $f$ under $\calD$ is at least the expected value of $f$ under a product distribution with the same element-marginals.

\medskip
Next, we give several applications of Submodular Dominance to submodular optimization. In particular, we improve the best known submodular prophet inequalities, we develop new rounding techniques for polytopes of set systems that admit negatively dependent distributions, and we prove existence of contention resolution schemes for WNR distributions.
\end{abstract}

\setcounter{tocdepth}{1}

{\small\tableofcontents}

\thispagestyle{empty}

\thispagestyle{empty}

\clearpage

\setcounter{page}{1}

\section{Introduction}

A function $f : 2^U \rightarrow \mathbb{R}$ on universe $U = \{1, \ldots, n\}$ is \emph{submodular} if it satisfies $f(S) + f(T) \geq f(S \cup T) + f(S \cap T)$ for all $S, T \subseteq U$. These functions capture the concept of diminishing returns, and are therefore useful in many fields such as machine learning, operations research, mechanism design, and combinatorial optimization; see books~\cite{Fujishige-Book05,Bach-Book13,Schrijver-Book03,NRTV2007}. 

Although $f$ is a discrete function, for many applications it is useful to define a continuous relaxation $f_{\text{cont}}: [0,1]^n \rightarrow \mathbb{R}$ of $f$, since that allows us to use techniques from continuous optimization. Here, by a relaxation we mean that $f_{\text{cont}}$ equals $f$ at the indicator vectors of the sets, i.e., $f_{\text{cont}}(\mathsf{1}_S) = f(S)$ for all $S \subseteq U$. A standard way to define such continuous relaxations is to first define a probability distribution $\calD(\mathbf{x})$ over subsets of $U$ with element-marginals $\mathbf{x} \in [0,1]^{n}$, and then define $f_{\text{cont}}(\mathbf{x})$ to be the expectation with respect to this distribution, i.e., $f_{\text{cont}}(\mathbf{x}) \coloneqq \E_{S \sim \calD(\mathbf{x})}[f(S)]$, where $S$ is a random set drawn from $\calD(\mathbf{x})$. For example, the popular \emph{multilinear relaxation} $F(\mathbf{x})$ is defined by taking $\calD(\mathbf{x})$ to be the product distribution with marginals $\mathbf{x}$. Other examples include the convex closure relaxation $f^-(\mathbf{x})$ (which is equivalent to the Lov\'asz extension for submodular functions), the concave closure relaxation $f^+(\mathbf{x})$, and the relaxation $f^*(\mathbf{x})$~\cite{Vondrak07}. Studying the properties of submodular expectations for these distributions has been a fruitful direction, which has led us to several optimal/approximation algorithms for submodular optimization~\cite{Bach-Book13,Vondrak-STOC08,CCPV-SICOMP11,FNS-FOCS11,AN15,EneN-FOCS16}.

Given the success of the above continuous relaxations, it is natural to ask what other continuous relaxations, or equivalently, what other submodular expectations and distributions $\calD(\mathbf{x})$ could be defined that are useful for new or improved applications. In this work, we study submodular expectations for negatively dependent distributions. Besides being of intellectual interest, we use them to improve the best known submodular prophet inequalities, to develop new rounding techniques, and to design contention resolution schemes for negatively dependent distributions.


\subsection{Submodular Dominance}

Since the multilinear extension $F$ is commonly employed in combinatorial optimization, one avenue to explore other continuous relaxations is by comparing them to $F$.

\begin{definition}[Submodular Dominance] \label{def:SubmodularDominance}
    A distribution $\calD$ over $2^U$ with marginals $\mathbf{x} \in [0, 1]^n$ satisfies \emph{Submodular Dominance} if for every submodular function $f : 2^U \rightarrow \mathbb{R}$,
    \[
        \E_{S \sim \calD}[f(S)] \quad \geq \quad F(\mathbf{x}) \enspace .
    \]
\end{definition}

Shao~\cite{Shao00} studied a similar concept that he called a comparison theorem, which involved a subclass of submodular functions. Christofides and Vaggelatou~\cite{CV04} later studied what they called the supermodular ordering, which is essentially equivalent to Submodular Dominance. Both viewed the problem through the lens of probability theory, whereas we approach it from the standpoint of combinatorial optimization.

It is not difficult to see how one might apply Submodular Dominance, e.g., it immediately yields an algorithm to round multilinear extension subject to feasibility constraints. However, Submodular Dominance implies a much wider variety of results in stochastic settings, where most of our current understanding relies on the independence of random variables. By relating product distributions to more complex distributions, Submodular Dominance allows us to improve existing results and study more general problems.


\subsection{Negative Dependence and Submodular Dominance}

Positive correlations can only decrease the expectations of submodular functions due to their diminishing marginal returns, so we turn our attention to negatively dependent distributions. Pemantle initiated a systematic study of such distributions in~\cite{Pemantle99}. In this work, we introduce the following generalization of Negative Association (NA) and Negative Regression (NR),  two popular notions of negative dependence (details in \Cref{sec:WNR}).

\begin{definition}[WNR] \label{def:WNR}
    A distribution $\calD$ over $2^U$ satisfies \emph{Weak Negative Regression} (WNR) if for any $i \in U$ and any monotone function $f : 2^U \rightarrow \mathbb{R}$,\footnote{A function $f$ is monotone if it satisfies $f(S) \leq f(T)$ for all $S \subseteq T$. Elements should be taken as singleton sets depending on context, e.g., $S \setminus i$ means $S \setminus \{i\}$.}
    \begin{align} \label{eq:WNRCondition}
        \E_{S \sim \calD}[f(S \setminus i) \given i \in S] \quad \leq \quad \E_{S \sim \calD}[f(S \setminus i) \given i \not\in S] \enspace .
    \end{align}
\end{definition}

Equivalently, $\calD$ is WNR if $S \setminus i$ conditioned on $i \not\in S$ stochastically dominates $S \setminus i$ conditioned on $i \in S$ for all $i \in U$. This captures an intuitive notion of negative dependence where conditioning on including an element lowers the probability of other inclusion events. WNR distributions satisfy Submodular Dominance as well as many desirable closure properties.

\medskip
\textbf{Submodular Dominance for Negatively Dependent Distributions.}
Christofides and Vag-gelatou~\cite{CV04} proved that NA distributions over continuous random variables satisfy Submodular Dominance for a continuous generalization of submodular functions. We  strengthen their result in \Cref{sec:CompIneq} in the setting of Bernoulli random variables from NA to WNR distributions, a strict superset of the union of NA and NR distributions.

\begin{theorem}[restate=CompIneqSufficient,name=] \label{thm:CompIneqSufficient}
    WNR distributions satisfy Submodular Dominance.
\end{theorem}

It turns out that there exist distributions that satisfy Submodular Dominance but are not WNR. This raises the question: what conditions are necessary for Submodular Dominance? We first recall the classic notion of Negative Cylinder Dependence (see, e.g.,~\cite{GV18}). 

\begin{definition}[NCD] \label{def:NCD}
    A distribution $\calD$ over $2^U$ with marginals $\mathbf{x}$ satisfies \emph{Negative Cylinder Dependence} (NCD) if for any $T \subseteq U$,\footnote{$S \sim \mathbf{x}$ means $S$ is sampled from the product distribution with marginals $\mathbf{x}$.}
    \[
        \Pr_{S \sim \calD}[T \subseteq S] \quad \leq \quad \Pr_{S \sim \mathbf{x}}[T \subseteq S] \qquad \text{and} \qquad \Pr_{S \sim \calD}[T \subseteq S^c] \quad \leq \quad \Pr_{S \sim \mathbf{x}}[T \subseteq S^c] \enspace .
    \]
\end{definition}

NCD can be interpreted as saying that any subset of elements are negatively correlated.

\begin{theorem}[restate=CompIneqNecessary,name=] \label{thm:CompIneqNecessary}
    All distributions that satisfy Submodular Dominance are NCD.
\end{theorem}

This can be useful when Submodular Dominance is an easier property to prove. For example, the distribution arising from randomized swap rounding can be shown to satisfy Submodular Dominance via a straightforward convexity argument, but a direct proof that the distribution is NCD is more involved~\cite{CVZ-arXiv09}; this theorem shows that such results follow due to a natural relationship between Submodular Dominance and negative dependence rather than any algorithm specific properties.

Although NCD is necessary for Submodular Dominance, it is insufficient on its own. While this insufficiency result was previously known~\cite{CVZ-FOCS10},\footnote{Observing that certain randomized rounding algorithms give rise to distributions satisfying both Submodular Dominance and NCD, Chekuri, Vondr{\'{a}}k, and Zenklusen~\cite{CVZ-FOCS10} remarked that there exist NCD distributions which violate Submodular Dominance, so NCD was not sufficient for Submodular Dominance. Our \Cref{thm:CompIneqNecessary} shows the other direction, that Submodular Dominance  implies NCD.} we strengthen it by constructing an example of an NCD distribution which violates Submodular Dominance and is additionally homogeneous, meaning it is distributed only on sets of the same size. Such distributions occur often enough to be of interest, e.g., distributions over the bases of a matroid.


\subsection{Applications}

Besides being a natural question, Submodular Dominance has several applications. 

\medskip
\textbf{Submodular Prophet Inequalities.}
The Prophet Inequality is a classical problem where a gambler sees the realizations of non-negative random variables one-by-one, choosing a random variable in an online fashion and attempting to maximize its value. The celebrated result of Krengel, Sucheston, and Garling~\cite{Krengel-Journal77,Krengel-Journal78} demonstrates a $\nicefrac{1}{2}$ \emph{prophet inequality}, meaning that just knowing the distributions in advance is enough to obtain $\nicefrac{1}{2}$ the expectation obtained by the \emph{prophet} that knows all the realizations in advance.

Motivated by applications to mechanism design, several works extended the $\nicefrac{1}{2}$ prophet inequality to gamblers selecting multiple random variables subject to a packing constraint to maximize a linear objective function, e.g.,~\cite{HKS-AAAI07,CHMS-STOC10,Alaei-FOCS11,KW-STOC12,Rubinstein-STOC16}. The \emph{Submodular Prophet Inequality} (SPI) was introduced by Rubinstein and Singla~\cite{RS-SODA17} as a further generalization to submodular objective functions to capture combinatorial applications.

One significant complication in SPI is that beyond simple Bernoulli settings, we deal with expectations that are no longer taken over product distributions. Chekuri and Livanos~\cite{CL21} obtain an efficient\footnote{We use efficient to mean algorithms that run in probabilistic polynomial time.} $c \cdot (1-e^{-b}) \cdot (e^{-b}-\epsilon)$ SPI for set systems with solvable polytopes\footnote{The polytope $\calP_\calI$ of a set system $\calI$ is formed by taking the convex hull of the indicator vectors of maximal independent sets in $\calI$, and is solvable if linear objective functions can be efficiently maximized over it.} and an efficient $(b, c)$-selectable greedy online contention resolution scheme (OCRS) for product distributions (see formal definitions in \Cref{sec:SPIApplication}). Crucially, their result loses a factor of $e^{-b}-\epsilon$ to handle the non-product distributions of SPI. We use Submodular Dominance to re-analyze the performance of greedy OCRSs in \Cref{sec:SPIImprovement}, which allows us to save this factor of $e^{-b}-\epsilon$ and improve the best known SPIs.

\begin{theorem}[restate=SPI,name=Submodular Prophet Inequalities] \label{thm:SPI}
    For fixed $\epsilon > 0$, if a set system $\calI \subseteq 2^U$ has a solvable polytope and an efficient $(b, c)$-selectable greedy OCRS for product distributions:
    \begin{itemize}[topsep=0pt,itemsep=0pt]
        \item There is an efficient $c \cdot (1-e^{-b}-\epsilon)$ SPI for monotone non-negative submodular functions.
        
        \item There is an efficient $\nicefrac{c}{4} \cdot (1-e^{-b}-\epsilon)$ SPI for general non-negative submodular functions.
    \end{itemize}
    Combining with known greedy OCRSs, this implies efficient SPIs as given in \Cref{table:SPINumbers}.
\end{theorem}

\begin{table}[ht]
    \begin{center}
    {\renewcommand{\arraystretch}{1.23} \normalfont
    \begin{tabular}{|C{5cm}|C{1.6cm}|C{1.6cm}|C{1.6cm}|C{1.6cm}|} \hline
        Feasibility Constraint                  & \multicolumn{2}{c|}{Prior Best~\cite{CL21}}
                                                & \multicolumn{2}{c|}{Our Results} \\ \cline{2-5}
                                                & Monotone & General  & Monotone & General  \\ \hline\hline
        Uniform Matroid of rank $k \rightarrow \infty$
                                                & $1/4.30$  & $1/17.20$ & $1-\nicefrac{1}{e}-\epsilon$
                                                                                    & $1/6.33$  \\ \hline
        Matroid                                 & $1/7.39$  & $1/29.54$ & $1/5.02$  & $1/20.07$ \\ \hline
        Matching                                & $1/9.49$  & $1/37.93$ & $1/6.75$  & $1/27.00$ \\ \hline
        Knapsack                                & $1/17.41$ & $1/69.64$ & $1/13.40$ & $1/53.60$ \\ \hline
    \end{tabular}
    }
    \end{center}
    \caption{Submodular Prophet Inequalities for different feasibility constraints. }
    \label{table:SPINumbers}
\end{table}

It is known that even for offline monotone submodular maximization over uniform matroids, no efficient algorithm can do better than a $(1-\nicefrac{1}{e})$-approximation~\cite{nemhauserWolsey78}. Thus, we obtain the first optimal efficient $1-\nicefrac{1}{e}-\epsilon$ monotone SPI over large rank uniform matroids.

\medskip
\textbf{Submodular Maximization.}
Another application is sampling from WNR distributions as a randomized rounding technique where the integral solution obtains at least the value of the fractional solution in expectation. A common method in submodular optimization is to first maximize the multilinear extension, which Vondr{\'{a}}k~\cite{Vondrak-STOC08} showed can be done for downward-closed set systems with solvable polytopes. For matroids, we know of methods which round the fractional solutions to sets without losing value \cite{CCPV07,CVZ-arXiv09,CVZ-FOCS10}, but set systems with solvable polytopes are far more general than matroids. Thus, the challenge in going beyond matroids is rounding the multilinear extension. By Submodular Dominance, an algorithm that efficiently generates a WNR distribution for a polytope automatically rounds the multilinear extension, which we show has immediate consequences for submodular maximization (details in \Cref{sec:SubmodMax}).

\begin{theorem}[restate=SubmodMax,name=Submodular Maximization] \label{thm:SubmodMax}
    Let $f : 2^U \rightarrow \mathbb{R}_{\geq 0}$ be a monotone submodular function. If a downward-closed set system $\calI \subseteq 2^U$ has a solvable polytope and efficiently admits WNR distributions, there exists an efficient algorithm that returns $T \in \calI$ such that $\E[f(T)] \geq (1-\nicefrac{1}{e}-o(1)) \cdot \max_{S \in \calI} f(S)$.
\end{theorem}

\medskip
\textbf{Adaptivity Gaps for Stochastic Probing.}
A natural generalization of submodular maximization is by adding stochasticity: replace elements by random variables called items. Such problems are often known as Stochastic Probing~\cite{GN-IPCO13,AN15,GNS-SODA16,BSZ-Random19,EKM-COLT21}. In addition to knowing the distributions of the items, we also allow algorithms to learn the realization of an item after selecting it. This opens up the concept of adaptive algorithms, which modify their behavior conditioned on such realizations. Though adaptivity can result in better algorithms, it also introduces significant complexity; for example, a decision tree can be of exponential size. Therefore, non-adaptive algorithms may be preferable if their performance is comparable to that of the optimal adaptive algorithm, a concept known as the adaptivity gap. By sampling from WNR distributions to round the multilinear extension, we adapt the analysis of the adaptivity gap upper bound by Asadpour and Nazerzadeh~\cite{AN15} from matroids to any set system for which WNR distributions exist (details in \Cref{sec:StocProbing}).

\begin{theorem}[restate=StocProbing,name=Stochastic Probing] \label{thm:StocProbing}
    For a downward-closed set system $\calI$ that admits WNR distributions, the adaptivity gap for Stochastic Probing is upper-bounded by $\frac{e}{e-1}$.
\end{theorem}

\medskip
\textbf{Contention Resolution Schemes.}
Contention resolution schemes (CRS) are another randomized rounding technique, with the concept being formally introduced by \cite{CVZ-SICOMP14} for submodular maximization. (Similar but less thoroughly explored notions appear in earlier works such as \cite{BKNS12}.) Since submodular maximization usually occurs via approximations of the multilinear extension, CRSs have generally been studied with respect to product distributions. Recently, Dughmi~\cite{Dughmi-ICALP20,Dughmi-arXiv21} initiated the study of CRSs for non-product distributions because of their applications in settings such as the Matroid Secretary Problem. We extend the CRS of \cite{CVZ-SICOMP14} for matroids from product distributions to WNR distributions, which gives possible directions to generalize our understanding of CRSs (details in \Cref{sec:CRS}).

\begin{theorem}[restate=CRS,name=Contention Resolution Schemes] \label{thm:CRS}
    For a matroid $\calM$, there exists a $(1-\nicefrac{1}{e})$-selectable CRS for any WNR distribution with marginals $\mathbf{x} \in \calP_\calM$.
\end{theorem}

\section{WNR and Other Negatively Dependent Distributions} \label{sec:WNR}

In this section, we first discuss  popular notions of negative dependence, and then introduce WNR and study its various  properties. Lengthier proofs are deferred to \Cref{sec:WNRProofs}.

\begin{definition}[NA] \label{def:NA}
    A distribution $\calD$ over $2^U$ satisfies \emph{Negative Association} (NA) if for any monotone $f, g : 2^U \rightarrow \mathbb{R}$ depending on disjoint sets of elements, $\Cov_{S \sim \calD}[f(S), g(S)] \leq 0$.
\end{definition}

This property is very similar to the Positive Association (PA) condition in the FKG inequality, with the main difference being the reversed inequality and disjoint sets. Although the FKG Inequality gives a straightforward condition to check for PA, no analogous result exists for NA, and in general, it is difficult to prove that a distribution is NA~\cite{JP83,Pemantle99}.

Another way to define negative dependence is based on the idea that conditioning on ``larger'' inclusion events should reduce the probability of other inclusion events.

\begin{definition}[NR] \label{def:NR}
    A distribution $\calD$ over $2^U$ satisfies \emph{Negative Regression} (NR) if for any sets $R_-, R_+, T \subseteq U$ such that $R_- \subseteq R_+ \subseteq T$ and any monotone function $f : 2^U \rightarrow \mathbb{R}$,
    \[
        \E_{S \sim \calD}[f(S \setminus T) \given (S \cap T) = R_+] \quad \leq \quad \E_{S \sim \calD}[f(S \setminus T) \given (S \cap T) = R_-] \enspace .
    \]
\end{definition}

Equivalently, $\calD$ is NR if $S \setminus T$ conditioned on $S \cap T = R_-$ stochastically dominates $S \setminus T$ conditioned on $S \cap T = R_+$ for all $R_- \subsetneq R_+ \subseteq T \subseteq U$. It turns out that NR is also a difficult property to check.

Since NA and NR are both natural forms of negative dependence, it is surprising that the exact relationship between them is unknown. While it is known that NA does not imply NR, it is conjectured that NR implies NA~\cite{Pemantle99}. One might then ask whether we can generalize NA and NR to get the best of both worlds: a weaker notion of negative dependence that is easier to check while still satisfying many desirable properties. This is our motivations for defining WNR distributions.

We first reformulate the WNR condition in terms of covariance.

\begin{claim} \label{claim:WNRCovariance}
    The WNR condition \eqref{eq:WNRCondition} is equivalent to $\Cov_{S \sim \calD}[f(S \setminus i), \mathbbm{1}_{i \in S}] \leq 0$.
\end{claim}

\begin{proof}
Let $\calD$ have marginals $\mathbf{x}$. The covariance inequality is equivalent to $\E_{S \sim \calD}[f(S \setminus i) \cdot \mathbbm{1}_{i \in S}] \leq \E_{S \sim \calD}[f(S \setminus i)] \E_{S \sim \calD}[\mathbbm{1}_{i \in S}]$. Then observe that if we expand LHS by conditioning on $i$, the expectation conditioned on $i \not\in S$ is $0$ due to the indicator. We can also simplify RHS using $\E_{S \sim \calD}[\mathbbm{1}_{i \in S}] = x_i$, so the covariance inequality is equivalent to
\[
    x_i \cdot \E_{S \sim \calD}[f(S \setminus i) \cdot \mathbbm{1}_{i \in S} \given i \in S] \quad \leq \quad \E_{S \sim \calD}[f(S \setminus i)] \cdot x_i \enspace .
\]
Canceling the $x_i$ and noting that the indicator on LHS is always $1$, we have $\E_{S \sim \calD}[f(S \setminus i) \given i \in S] \leq \E_{S \sim \calD}[f(S \setminus i)]$. This is equivalent to the WNR condition since the expectation is just a weighted sum of the conditional expectation on $i \in S$ and $i \not\in S$.
\end{proof}

Next, we prove that WNR distributions generalize NA and NR distributions.
\begin{proposition}[restate=WNRisSuperset,name=]
    NA and NR imply WNR, and WNR implies NCD, but the reverse implications do not hold. In other words, the union of NA and NR distributions is a strict subset of WNR distributions, which is a strict subset of NCD distributions.
\end{proposition}

\begin{proof}
The WNR condition is a special case of the NR condition when $T \coloneqq \{i\}$, and by \Cref{claim:WNRCovariance}, it is also a special case of the NA condition when $g(S) \coloneqq \mathbbm{1}_{i \in S}$, so both NA and NR imply WNR. For strict containment, we give a WNR distribution that is neither NA nor NR in \Cref{sec:WNRProofs}. \Cref{thm:CompIneqSufficient,thm:CompIneqNecessary} and the example distributions in \Cref{sec:CompIneq} demonstrate that WNR distributions are a strict subset of NCD distributions.
\end{proof}

\begin{figure}[ht]
\begin{center}
\begin{tikzpicture}[thin,scale=0.6]
	\draw [fill=blue!10]   (1.5,0) ellipse (6.7cm and 3.5cm);
	\draw [fill=green!14]  (1.5,-0.4) ellipse (5.8cm and 2.7cm);
 	\draw [fill=yellow!18] (1.5,-0.8) ellipse (4.9cm and 1.9cm);
 	\draw [fill=red!26]    (2.75,-1.2) ellipse (2.75cm and 1cm);
	\draw [fill=red!26]    (0.25,-1.2) ellipse (2.75cm and 1cm);

	\draw (0.25,-1.2) ellipse (2.75cm and 1cm);
	\draw (2.75,-1.2) ellipse (2.75cm and 1cm);

 	\node at (-0.8,-2.2)[label={[align=center]Negative \\ Association}] {};
 	\node at (3.8,-2.3)[label={[align=center]Negative \\ Regression}] {};
 	\node at (1.5,1.9)[label=Negative Cylinder Dependence] {};		
 	\node at (1.5,0.9)[label=Submodular Dominance] {};		
 	\node at (1.5,-0.5)[label=Weak Negative Regression] {};		
\end{tikzpicture}
\end{center}
\caption{Hierarchy of negative dependence and its relation to Submodular Dominance.}
\label{fig:PackingConstrHierarchy}
\end{figure}
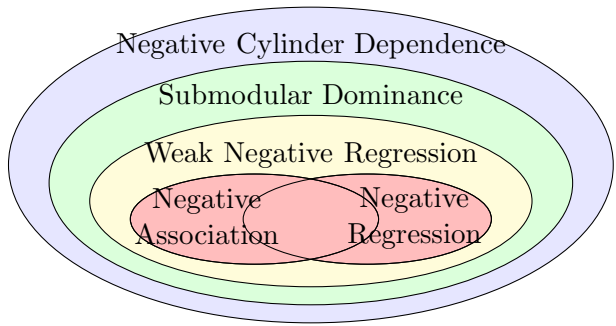

Finally, we observe that WNR satisfies two closure properties, proved in \Cref{sec:WNRProofs}.

\begin{definition}[Projection]
    Let $\calD$ be a distribution over $2^U$. Its \emph{projection} onto $U' \subseteq U$ is the distribution which samples $S \sim \calD$ and returns $S \cap U'$.
\end{definition}

\begin{definition}[Products]
    Let $\calA$ and $\calB$ be distributions over $2^A$ and $2^B$ for disjoint $A, B$. Their \emph{product}  distribution  independently samples $S \sim \calA$ and $T \sim \calB$, then returns $S \cup T$.
\end{definition}

\begin{proposition}[restate=WNRProperties,name=] \label{prop:WNRProperties}
    WNR is closed both under projection and  under products.
\end{proposition}

Joag-Dev and Proschan~\cite{JP83} showed that NA distributions are NCD, closed under projection, closed under products, and closed under taking monotone functions of disjoint subsets of variables.\footnote{This last property is useful in the setting of continuous random variables studied in~\cite{JP83} because properties closed under convolutions are extremely powerful. However, it is not so relevant in the discrete settings we study.} Since WNR shares three of these properties with NA and requires a weaker condition while also generalizing NR, it appears to be a useful notion of negative dependence.

\section{Towards a Characterization of Submodular Dominance} \label{sec:CompIneq}


\subsection{WNR is a Sufficient Condition}

\CompIneqSufficient*

\begin{proof}
We prove by induction on the number of elements. The base case of $1$ element is trivial because the marginals fully specify the distribution.

We will show that any WNR distribution $\calD$ over $2^{[k]}$ with marginals $\mathbf{x}$ satisfies Submodular Dominance, assuming by induction that all WNR distributions over $2^{[k-1]}$ satisfy Submodular Dominance. We assume that $x_k \ne 0, 1$ because otherwise, we can interpret $\calD$ as a distribution over $2^{[k-1]}$ and trivially be done.

Let $\calD \setminus k$ and $\mathbf{x} \setminus k$ denote the projections of $\calD$ and $\mathbf{x}$ onto $[k-1]$. Let $\calD_k$ denote the distribution which samples $S \sim \calD \setminus k$, then returns $S \cup k$ w.p. $x_k$ and returns $S$ otherwise, i.e., $\calD_k$ is $\calD$ but with element $k$ sampled independently.

Let $f : 2^{[k]} \rightarrow \mathbb{R}$ be a submodular function. To prove Submodular Dominance, we will show the following inequalities hold:
\begin{equation} \label{eq:CompIneqSteps}
    \E_{S \sim \calD}[f(S)] \quad \stackrel{\text{Claim~\ref{claim:CompIneqStep2}}}{\geq} \quad \E_{S \sim \calD_k}[f(S)] \quad \stackrel{\text{Claim~\ref{claim:CompIneqStep1}}}{\geq} \quad \E_{S \sim \mathbf{x}}[f(S)] \enspace .
\end{equation}

\begin{claim} \label{claim:CompIneqStep1}
    The second inequality of \eqref{eq:CompIneqSteps} holds, i.e., $\E_{S \sim \calD_k}[f(S)] \geq \E_{S \sim \mathbf{x}}[f(S)]$.
\end{claim}

\begin{proof}
Since $k$ is independently sampled in both $\calD_k$ and $\mathbf{x}$, we can write
\begin{align*}
    \E_{S \sim \calD_k}[f(S)] \quad &= \quad \E_{S \sim \calD \setminus k}[x_k \cdot f(S \cup k) + (1-x_k) \cdot f(S)] \enspace  \qquad \text{and}\\
    \E_{S \sim \mathbf{x}}[f(S)] \quad &= \quad \E_{S \sim \mathbf{x} \setminus k}[x_k \cdot f(S \cup k) + (1-x_k) \cdot f(S)] \enspace .
\end{align*}
Convex combinations of submodular functions are submodular, $\calD \setminus k$ is WNR by closure under projection (\Cref{prop:WNRProperties}), and the marginals of $\calD \setminus k$ are equal to the marginals of $\mathbf{x} \setminus k$. Therefore, by the induction hypothesis,
\[
    \E_{S \sim \calD \setminus k}[x_k \cdot f(S \cup k) + (1-x_k) \cdot f(S)] \quad \geq \quad \E_{S \sim \mathbf{x} \setminus k}[x_k \cdot f(S \cup k) + (1-x_k) \cdot f(S)] \enspace ,
\]
which implies $\E_{S \sim \calD_k}[f(S)] \geq \E_{S \sim \mathbf{x}}[f(S)]$.
\end{proof}

\begin{claim} \label{claim:CompIneqStep2}
    The first inequality of \eqref{eq:CompIneqSteps} holds, i.e., $\E_{S \sim \calD}[f(S)] \geq  \E_{S \sim \calD_k}[f(S)]$.
\end{claim}

\begin{proof}
Expanding expectations for $\calD_k$ (as in the proof of \Cref{claim:CompIneqStep1}) and moving it to LHS,
\[
    \E_{S \sim \calD}[f(S) - x_k \cdot f(S \cup k) - (1-x_k) \cdot f(S \setminus k)] \quad \geq \quad 0 \enspace .
\]
Conditioning on $k$ yields
\begin{align*}
    & \quad\quad\quad\;\; x_k \cdot \E_{S \sim \calD}[f(S \cup k) - x_k \cdot f(S \cup k) - (1-x_k) \cdot f(S \setminus k) \given k \in S] \\
    & + (1-x_k) \cdot \E_{S \sim \calD}[f(S \setminus k) - x_k \cdot f(S \cup k) - (1-x_k) \cdot f(S \setminus k) \given k \not\in S] \quad \geq \quad 0 \enspace ,
\end{align*}
which simplifies to
\[
    x_k(1-x_k) \Big(\E_{S \sim \calD}[f(S \cup k) - f(S \setminus k) \given k \in S] + \E_{S \sim \calD}[f(S \setminus k) - f(S \cup k) \given k \not\in S]\Big) \quad \geq \quad 0 \enspace .
\]
Dividing out $x_k (1-x_k)$ and moving the second term to RHS gives
\begin{equation} \label{eq:CompIneqStep2}
    \E_{S \sim \calD}[f(S \cup k) - f(S \setminus k) \given k \in S] \quad \geq \quad \E_{S \sim \calD}[f(S \cup k) - f(S \setminus k) \given k \not\in S] \enspace .
\end{equation}
Let $f_k(S) \coloneqq f(S \cup k) - f(S \setminus k)$. $f_k$ does not depend on $k$, and by the submodularity of $f$, $-f_k$ is a monotone function. Thus, \eqref{eq:CompIneqStep2} is directly implied by the WNR condition \eqref{eq:WNRCondition}.
\end{proof}

\Cref{claim:CompIneqStep1,claim:CompIneqStep2} complete the proof of \Cref{thm:CompIneqSufficient}.
\end{proof}

The following proposition, which we prove in \Cref{sec:NotTightProofs}, shows that WNR is not a necessary condition for Submodular Dominance.

\begin{proposition}[restate=CompIneqUnnecessary,name=]
    The distribution $\calD$ which samples uniformly from $\emptyset, \{1\}, \{2\}, \{1, 2\}, \{1, 3\}$, $\{2, 3\}$ satisfies Submodular Dominance, but $\calD$ violates WNR for $f(S) \coloneqq \max(\mathbbm{1}_{1 \in S}, \mathbbm{1}_{2 \in S})$ and $i = 3$.
\end{proposition}


\subsection{NCD is a Necessary Condition}

Since WNR is not equivalent to Submodular Dominance, we search for necessary conditions to better understand the relationship between negative dependence and Submodular Dominance.

\CompIneqNecessary*

\begin{proof}
Let $\calD$ be a distribution over $2^U$ with marginals $\mathbf{x}$ which satisfies Submodular Dominance. For any $T \subseteq U$, consider the functions
\[
    f_T(S) \quad \coloneqq \quad 1 - \mathbbm{1}_{T \subseteq S^c} \qquad \text{and} \qquad g_T(S) \quad \coloneqq \quad \vert S \cap T \vert - \mathbbm{1}_{T \subseteq S} \enspace .
\]
Equivalently, $f_T, g_T$ are the rank functions of the uniform matroids of rank $1$ and $\vert T \vert - 1$ over the ground set $T$. The only fact about matroid rank functions we use here is that all matroid rank functions are submodular (though one can also easily check that $f_T, g_T$ are submodular via the definition of submodularity). Since subtracting a linear function from a submodular function results in a submodular function, $-\mathbbm{1}_{T \subseteq S}$ and $-\mathbbm{1}_{T \subseteq S^c}$ are submodular functions. Thus, by Submodular Dominance we have
\begin{align*}
    -\Pr_{S \sim \calD}[T \subseteq S] \quad\, = \,\quad \E_{S \sim \calD}[-\mathbbm{1}_{T \subseteq S}] \quad\, &\geq \,\quad \E_{S \sim \mathbf{x}}[-\mathbbm{1}_{T \subseteq S}] \quad\, = \,\quad -\Pr_{S \sim \mathbf{x}}[T \subseteq S] \qquad \text{and} \\
    -\Pr_{S \sim \calD}[T \subseteq S^c] \quad = \quad \E_{S \sim \calD}[-\mathbbm{1}_{T \subseteq S^c}] \quad &\geq \quad \E_{S \sim \mathbf{x}}[-\mathbbm{1}_{T \subseteq S^c}] \quad = \quad -\Pr_{S \sim \mathbf{x}}[T \subseteq S^c] \enspace .
\end{align*}
Multiplying both sides by $-1$ yields the definition of NCD (\Cref{def:NCD}).
\end{proof}

The following two propositions, which we prove in \Cref{sec:NotTightProofs}, show that NCD is not a sufficient condition for Submodular Dominance.

\begin{proposition}[restate=CompIneqInsufficient,name=] \label{prop:NCDCounterexample}
    The distribution $\calD$ over $2^{[4]}$ which chooses uniformly at random $i \in [4]$, then returns w.p. $\nicefrac{1}{2}$ either $i$ or $[4] \setminus i$, is NCD. However, $\calD$ violates Submodular Dominance for the submodular function $f(S) \coloneqq \min(2, \vert S \vert)$.
\end{proposition}

\begin{proposition}[restate=CompIneqInsufficientH,name=]
    The distribution $\calD$ over $2^{[8]}$ which chooses uniformly at random $i \in A \coloneqq \{1, 2, 3, 4\}$ and $j \in B \coloneqq \{5, 6, 7, 8\}$, then returns w.p. $\nicefrac{1}{2}$ either $i \cup (B \setminus j)$ or $(A \setminus i) \cup j$, is NCD. However, $\calD$ violates Submodular Dominance for the submodular function $f(S) \coloneqq \min(2, \vert S \cap A \vert)$.
\end{proposition}

Thus, the class of distributions which satisfy Submodular Dominance is a strict subset of NCD distributions and a strict superset of WNR distributions. It is unclear whether the ``right'' answer will turn out to be a useful notion of negative dependence.

\section{Applications to Submodular Prophet Inequalities} \label{sec:SPIApplication}

In SPI, we have \emph{items} $U$, which are discrete random variables with disjoint images and arbitrary probability mass functions. We denote realizations of items as \emph{elements}. WLOG, let the image of $i \in U$ be $\{ij : j \in [m]\}$, and let the realization of $i$ be $ij$ w.p. $p_{ij}$. Let $E \coloneqq [n] \times [m]$ denote the set of elements. The distributions of each item are independent and known to us in advance.

We are given a set system $\calI \subseteq 2^U$ and a submodular objective function $f : 2^E \rightarrow \mathbb{R}_{\geq 0}$. Notice that while the items are independent, the elements do not follow a product distribution. As we are optimizing over the element-space, this is a non-trivial complication.

Each item arrives one-by-one. When an item arrives, we learn its realization, and must choose whether to accept or reject it. The set of accepted items must be in $\calI$, and the goal is to maximize $f$ on the realizations of the accepted items. The arrival order is chosen by an \emph{almighty adversary}, who knows in advance the outcomes of all randomness, such as the item realizations, the decisions of our algorithm, etc.

If there exists an $\alpha$-competitive algorithm compared to the prophet, we say there is an $\alpha$ SPI. Rubinstein and Singla~\cite{RS-SODA17} proved $\Omega(1)$ SPIs over matroids, and Chekuri and
Livanos~\cite{CL21} refined their analysis to obtain better constants, as well as results for a broader range of set systems. We further improve upon their approach using Submodular Dominance, obtaining results such as the first tight SPI for large rank uniform matroids.


\subsection{Core Approach: SPI for Bernoulli Items}

Before tackling the full problem, it is helpful to first consider a simplified version, \emph{Bernoulli SPI}, where each item $i$ is only a Bernoulli random variable, taking value $1$ w.p. $p_i$ and taking value $0$ otherwise. Here, there is no notion of elements (or rather, elements are effectively synonymous with items), so we consider a submodular objective function $f : 2^U \rightarrow \mathbb{R}_{\geq 0}$. As in the full problem, we have a set system constraint $\calI \subseteq 2^U$.

This is quite similar to the problem of submodular maximization from \Cref{sec:SubmodMax}, but with a stochastic component (each item being usable only w.p. $p_i$) and an online component (items are revealed one-by-one). Therefore, it makes sense to borrow the high level approach of optimizing the multilinear extension $F$, then rounding the fractional solution. Since $\calI$ is a discrete constraint and $F$ is a continuous function, the following relaxation is useful in offline submodular maximization:

\begin{definition}
    For any downward-closed set system $\calI \subseteq 2^U$, its \emph{polytope} $\calP_\calI \subseteq [0, 1]^n$ is the convex hull of the indicator vectors representing the maximal sets of $\calI$.
\end{definition}

For Bernoulli SPI, we consider a modified polytope $\calP'_\calI \coloneqq \{\mathbf{x} \in [0, 1]^n : \mathbf{x} \in \calP_\calI,\; x_i \leq p_i \; \forall i \in U\}$. Since the fractional solution corresponds to a distribution over $\calI$, the additional constraint $x_i \leq p_i$ ensures that no item is included more often than it takes value $1$. It turns out that we can efficiently optimize $F$ over $\calP'_\calI$ under mild conditions.

As for rounding, we can use \emph{online contention resolution schemes} (OCRS). OCRSs function in the following setting: we have a set system $\calI \subseteq 2^U$ and a distribution $\calD$ over $2^U$ with marginals $\mathbf{x}$. Let items $i \in S$ for some $S \sim \calD$ be called \emph{active}. The items then arrive one-by-one in adversarial order. When item $i$ arrives, we learn whether it is active, and if so, must decide to accept or reject it, subject to the set of accepted items being in $\calI$. An OCRS $\pi_{\calI, \calD}$ is an algorithm that plays this game. The following notion is a way to measure the performance of an OCRS.

\begin{definition}[$(b, c)$-selectable OCRS] \label{def:OCRSSelectability}
    For $b, c \in [0, 1]$, a set system $\calI \subseteq 2^U$, and a distribution $\calD$ over $2^U$ with marginals $\mathbf{x} \in b \cdot \calP_\calI$, an OCRS $\pi_{\calI, \calD}$ is \emph{$(b, c)$-selectable} if the probability of $\pi_{\calI, \calD}$ accepting $i$ is at least $c \cdot x_i$ for all $i \in U$. If $b = 1$, we say $\pi_{\calI, \calD}$ is $c$-selectable.
\end{definition}

Feldman, Svensson, and Zenklusen~\cite{FSZ15} obtained the following approximation result for rounding via \emph{greedy OCRSs} (we omit the definition as it is not relevant).

\begin{proposition}[\cite{FSZ15}] \label{prop:OCRSApproxProduct}
    For a set system $\calI \subseteq 2^U$, a monotone submodular function $f : 2^U \rightarrow \mathbb{R}_{\geq 0}$, and $\mathbf{x} \in b \cdot \calP_\calI$, applying a $(b, c)$-selectable greedy OCRS to $S \sim \mathbf{x}$ obtains $T \in \calI$ such that $\E_{S \sim \mathbf{x}}[f(T)] \geq c \cdot F(\mathbf{x})$. Further, the greedy OCRS can be efficiently modified such that even for non-monotone $f$, the modified greedy OCRS obtains $T \in \calI$ such that $\E_{S \sim \mathbf{x}}[f(T)] \geq \nicefrac{c}{4} \cdot F(\mathbf{x})$.
\end{proposition}

In Bernoulli SPI, there is no notion of elements and we optimize over items. Thus, the distribution of active items is already a product distribution, and simply applying greedy OCRSs already yields approximation results using \Cref{prop:OCRSApproxProduct}.


\subsection{Generalizing to Arbitrary Discrete Random Variables}

We now return to the full version of SPI. Again, let $U$ be the set of items, $\calI \subseteq 2^U$ be a downward-closed set system with a solvable polytope, $E \coloneqq [n] \times [m]$ be the set of elements, $\mathbf{p} \in [0, 1]^{nm}$ be the element realization probabilities, and $f : 2^E \rightarrow \mathbb{R}_{\geq 0}$ be a submodular function. The first step is to compute a fractional solution. Chekuri and Livanos~\cite{CL21} define the polytope
\[
    \calP''_\calI \quad \coloneqq \quad \{\mathbf{x} \in [0, 1]^{nm} \;:\; \exists \mathbf{z} \in \calP_\calI \text{~satisfying~}{\textstyle \sum_j} x_{ij} = z_i \;\; \forall i \in U, \;\;\, x_{ij} \leq p_{ij} \;\; \forall ij \in E\} \enspace .
\]
Here, the summation constraint is a natural relaxation of $\calP'_\calI$ from the item-space to the element-space. Chekuri and Livanos prove a series of results\footnote{See Section 3 of~\cite{CL21} for details, in particular, Claim 3.4, Theorem 1.3, and Remark 3.7.} which culminates in the following (note that we do not require monotonicity of $f$):

\begin{proposition}[\cite{CL21}] \label{prop:SPICorGap}
    Let $\OPT$ be the expectation obtained by the prophet, $\calI \subseteq 2^U$ be a set system with a solvable polytope, and $f : 2^U \rightarrow \mathbb{R}_{\geq 0}$ be a submodular function. Then for any fixed $\epsilon > 0$, we can efficiently compute $\mathbf{x} \in b \cdot \calP''_\calI$ such that $F(\mathbf{x}) \geq (1-e^{-b}-\epsilon) \cdot \OPT$.
\end{proposition}

It remains to round the fractional solution. While it is fairly straightforward to convert a $(b, c)$-selectable greedy OCRS for the item-space to a $(b, c)$-selectable greedy OCRS for the element-space (indeed, we do this in Algorithm~\ref{1}), we cannot obtain approximation results directly from \Cref{prop:OCRSApproxProduct} like in the Bernoulli case because the distribution of elements is not a product distribution. Chekuri and Livanos handle this by incurring an additional loss of $e^{-b}-\epsilon$ to ``mask'' the elements under a product distribution. We save this factor by re-analyzing a simpler algorithm.


\subsection{Improved Analysis} \label{sec:SPIImprovement}

Let $\mathbf{x} \in b \cdot \calP''_\calI$ be the solution computed as per \Cref{prop:SPICorGap}. Define $x_i \coloneqq \sum_{j} x_{ij}$, and define $\vec{x} \coloneqq (x_i : i \in U)$. By definition of $\calP''_\calI$ and the fact that $\mathbf{x} \in b \cdot \calP''_\calI$, we have that $\vec{x} \in b \cdot \calP_\calI$, so let $\pi_{\calI, \vec{x}}$ be an efficient $(b, c)$-selectable greedy OCRS.

We first consider monotone $f$. Our rounding algorithm is almost identical to \cite[Algorithm 1]{CL21}, removing some steps that our improved analysis demonstrates to be unnecessary.

\namelabel{1}
\begin{algorithm}
\SetAlgoLined
    $T_{\ALG} = \emptyset$ \\
    \For{$t \gets 1$ \KwTo $n$}{
        Let $i \in U$ be the item that arrives on day $t$ \\
        Let $ij \in E$ be the realization of $i$ \\
        With probability $x_{ij}/p_{ij}$, reveal active $i$ to $\pi_{\calI, \vec{x}}$, otherwise reveal inactive $i$ to $\pi_{\calI, \vec{x}}$ \\
        \If{$\pi_{\calI, \vec{x}}$ accepts $i$}{
            $T_{\ALG} \leftarrow T_{\ALG} \cup \{ij\}$
        }
    }
    Return $T_{\ALG}$
    \caption{\textsc{Monotone Rounding ($U, E, \mathbf{p}, f, \mathbf{x}, \pi_{\calI, \vec{x}}$)}}
\end{algorithm}

Denote element $ij$ as \emph{active} when $ij$ is the realization of $i$, and Algorithm~\ref{1} reveals active $i$ to $\pi_{\calI, \vec{x}}$. Since the elements do not follow a product distribution, we cannot apply \Cref{prop:OCRSApproxProduct} even though the algorithm acts like a greedy OCRS. However, we provide a new analysis which states that a $c$-approximation for product distributions implies a $c$-approximation for the following wider class of distributions.

\begin{definition}
    A \emph{product of singletons distribution} over $2^E$ with marginals $\mathbf{x} \in [0, 1]^{nm}$ such that $\sum_j x_{ij} \leq 1$ for all $i$ is a distribution which independently samples 0 or 1 elements from each set $\{ij : j \in [m]\}$ according to the marginals $\mathbf{x}$.
\end{definition}

It is not difficult to see that the active elements follow a product of singletons distribution with marginals $\mathbf{x}$. The following lemma, which we prove in \Cref{sec:POSasProduct}, draws a connection between product of singletons distributions and product distributions.

\begin{lemma}[restate=POSasProduct,name=] \label{lem:POSasProduct}
    Let $\calD$ be a product of singletons distribution over $2^E$ with marginals $\mathbf{x} \in [0, 1]^{nm}$. Let $x_i \coloneqq \sum_j x_{ij}$, and let $\vec{x} \coloneqq (x_i : i \in U)$. For any $\mathbf{u} \in [m]^n$, let $E_\mathbf{u} \coloneqq \{iu_i : i \in U\}$ and let $\calD_\mathbf{u}$ be a product distribution over $2^{E_\mathbf{u}}$ with marginals $\vec{x}$. Then for any  $g : 2^E \rightarrow \mathbb{R}$,
    \begin{align} \label{eq:POSasProduct}
         \E_{S \sim \calD}[g(S)] \quad = \quad \sum_{\mathbf{u} \in [m]^n} \bigg(\E_{S \sim \calD_\mathbf{u}}[g(S)] \cdot \prod_{i \in U} \frac{x_{iu_i}}{x_i}\bigg) \enspace .
    \end{align}
\end{lemma}

In simpler terms, $\calD_\mathbf{u}$ is the distribution which samples $S \sim \calD$, then replaces each element $ij \in S$ with the element $iu_i$. \Cref{lem:POSasProduct} states that any product of singletons distribution with marginals $\mathbf{x}$ can be written as a convex combination of product distributions with marginals $\vec{x}$.

\begin{lemma} \label{lem:SPIApproxMonotone}
    For monotone $f$, Algorithm~\ref{1} returns $T_{\ALG}$ such that $\E[f(T_{\ALG})] \geq c \cdot F(\mathbf{x})$.
\end{lemma}

\begin{proof}
Let $\calD$ be the distribution of active elements. While the adversary sees the item realizations and which items Algorithm~\ref{1} will reveal as active to $\pi_{\calI, \vec{x}}$, the adversary cannot influence $\calD$ because the decisions to reveal active $i$ do not depend on the item ordering.

Therefore, it is valid for us to ``partition'' the outcomes of randomness contributing to $\calD$. Since $\calD$ is a product of singletons distribution with marginals $\mathbf{x}$, \Cref{lem:POSasProduct} tells us that there exists a partition such that each part is a product distribution $\calD_\mathbf{u}$ over $2^{E_\mathbf{u}}$ with marginals $\vec{x}$. We fix some $\mathbf{u} \in [m]^n$ and analyze the performance of Algorithm~\ref{1} on the subset of randomness corresponding to the distribution $\calD_\mathbf{u}$.

Since $\calD_\mathbf{u}$ is a product distribution, $\vec{x} \in b \cdot \calP_\calI$, and the algorithm copies the acceptances of $\pi_{\calI, \vec{x}}$, Algorithm~\ref{1} acts exactly like a $(b, c)$-selectable greedy OCRS over $\calD_\mathbf{u}$. Most importantly, $\calD_\mathbf{u}$ being product distribution means we can apply \Cref{prop:OCRSApproxProduct} to get
\[
    \E_{S \sim \calD_\mathbf{u}}[f(T_{\ALG})] \quad \geq \quad c \cdot \E_{S \sim \calD_\mathbf{u}}[f(S)] \enspace .
\]
$T_{\ALG}$ is implicitly a randomized function of $S$, so we can rewrite LHS as
\[
    \E_{S \sim \calD_\mathbf{u}}\Big[\E[f(T_{\ALG}) \given S = S']\Big] \quad \geq \quad c \cdot \E_{S \sim \calD_\mathbf{u}}[f(S)] \enspace ,
\]
where the inner expectation is taken over the possible randomization of the underlying greedy OCRS $\pi_{\calI, \vec{x}}$ and the adversarial ordering of the items. As this holds for any $\mathbf{u}$, weighting the inequality and summing over all $\mathbf{u} \in [m]^n$ yields
\[
    \sum_{\mathbf{u} \in [m]^n} \bigg(\E_{S \sim \calD_\mathbf{u}}\Big[\E[f(T_{\ALG}) \given S = S']\Big] \cdot \prod_{i \in U} \frac{x_{iu_i}}{x_i}\bigg) \quad \geq \quad \sum_{\mathbf{u} \in [m]^n} \bigg(c \cdot \E_{S \sim \calD_\mathbf{u}}[f(S)] \cdot \prod_{i \in U} \frac{x_{iu_i}}{x_i}\bigg) \enspace .
\]
Factoring out the $c$ on RHS, then applying \Cref{lem:POSasProduct} to the functions $\E[f(T_{\ALG}) \given S = S']$ and $f(S)$ simplifies to $\E_{S \sim \calD}[f(T_{\ALG})] \geq c \cdot \E_{S \sim \calD}[f(S)]$.

A distribution which samples only sets of size $0$ or $1$ is WNR because conditioning on inclusion of an element excludes all other elements. Further, products of WNR distributions are WNR (\Cref{prop:WNRProperties}). Thus, product of singletons distributions are WNR, and applying Submodular Dominance (\Cref{thm:CompIneqSufficient}) on $\calD$ gives the following and completes the proof:
\[
    \E_{S \sim \calD}[f(T_{\ALG})] \quad \geq \quad c \cdot \E_{S \sim \calD}[f(S)] \quad \geq \quad c \cdot F(\mathbf{x}) \enspace . \qedhere
\]
\end{proof}

\begin{remark} \label{rem:SPIApproxGeneral}
    For general $f$, we can replace the greedy OCRS $\pi_{\calI, \vec{x}}$ by its efficient modification mentioned in \Cref{prop:OCRSApproxProduct}, then just repeat the proof of \Cref{lem:SPIApproxMonotone}. We lose an additional factor of $\nicefrac{1}{4}$ when we invoke \Cref{prop:OCRSApproxProduct} on the modified greedy OCRS, which gives us a $\nicefrac{c}{4}$-approximation algorithm when $f$ is not monotone.
\end{remark}

\SPI*

\begin{proof}
Combining \Cref{prop:SPICorGap} and \Cref{lem:SPIApproxMonotone} gives us a $c \cdot (1-e^{-b}-\epsilon)$ SPI for monotone non-negative submodular functions, and, as noted in \Cref{rem:SPIApproxGeneral}, a similar argument gives us a $\nicefrac{c}{4} \cdot (1-e^{-b}-\epsilon)$ SPI for general non-negative submodular functions. From~\cite{FSZ15}, we can efficiently construct greedy OCRSs satisfying the following properties:
\begin{itemize}[topsep=0pt,itemsep=0pt]
    \item $(b, 1-b)$-selectable over matroids, for $b \in [0, 1]$.
    
    \item $(b, e^{-2b})$-selectable over matchings, for $b \in [0, 1]$.
    
    \item $(b, \frac{1-2b}{2-2b})$-selectable over knapsacks, for $b \in [0, \nicefrac{1}{2}]$.
    
    \item $(1-o(1))$-selectable over uniform matroids of rank $k \rightarrow \infty$ \cite{CL21}.
\end{itemize}
To obtain the results in \Cref{table:SPINumbers}, we simply choose $b$ which maximizes $c \cdot (1-e^{-b}-\epsilon)$.
\end{proof}

\section{Applications to Rounding} \label{sec:Rounding}

A common problem setting is optimization constrained to some feasible set system $\calI$.

\begin{definition}
    For a downward-closed set system $\calI \subseteq 2^U$, its \emph{polytope} $\calP_\calI \subseteq [0, 1]^n$ is the convex hull of the indicator vectors representing the maximal sets of $\calI$.
\end{definition}

Under mild conditions, we can efficiently optimize over the polytope, then round the fractional solution $\mathbf{x} \in \calP_\calI$ to an integral solution $S \in \calI$. It is natural to think of $\mathbf{x}$ as a distribution over $\calI$ with those marginals. If these distributions exhibit certain properties, then sampling can be an effective rounding technique. We give results for set systems which satisfy the following property:

\begin{definition}
    A set system $\calI \subseteq 2^U$ \emph{admits WNR distributions} if for any $\mathbf{x} \in \calP_\calI$, there exists a WNR distribution over $\calI$ with marginals $\mathbf{x}$. If we can efficiently sample from these distributions, we say $\calI$ \emph{efficiently admits WNR distributions}.
\end{definition}


\subsection{Submodular Maximization} \label{sec:SubmodMax}

For a set system $\calI \subseteq 2^U$ and a monotone submodular function $f : 2^U \rightarrow \mathbb{R}_{\geq 0}$, a classical optimization problem is to efficiently find $T \in \calI$ such that $f(T)$ is a good approximation of $\max_{S \in \calI} f(S)$. We start by optimizing over $\calP_\calI$.

\begin{proposition}[\cite{Vondrak-STOC08}] \label{prop:ContGreedy}
    For any set system $\calI$ with a solvable polytope, we can efficiently compute $\mathbf{x} \in \calP_\calI$ such that $F(\mathbf{x}) \geq (1-\nicefrac{1}{e}-o(1)) \cdot \max_{S \in \calI} f(S)$.
\end{proposition}

Now, we want to round $\mathbf{x}$ to an integral solution without losing value compared to $F(\mathbf{x})$. Pipage rounding~\cite{CCPV07} and randomized swap rounding~\cite{CVZ-arXiv09} achieve this for matroid polytopes, but it is unclear how to extend it. Submodular Dominance gives new approaches for submodular maximization.

\SubmodMax*

\begin{proof}
We compute $\mathbf{x} \in \calP_\calI$ as per \Cref{prop:ContGreedy}, then sample from a WNR distribution over $\calI$ with marginals $\mathbf{x}$. By \Cref{thm:CompIneqSufficient}, this returns $T \in \calI$ such that $\E[f(T)] \geq F(\mathbf{x}) \geq (1-\nicefrac{1}{e}-o(1)) \cdot \max_{S \in \calI} f(S)$.
\end{proof}


\subsection{Adaptivity Gaps for Stochastic Probing} \label{sec:StocProbing}

\emph{Stochastic Probing} is a generalization of submodular maximization with randomized inputs. Elements are replaced by \emph{items}, and we \emph{probe} items to learn their realizations. The goal is to maximize the expectation of a function over the realizations of probed items. We consider a simple version of the problem where we have a monotone submodular function $f : 2^U \rightarrow \mathbb{R}_{\geq 0}$ and each item $X_i$ contains element $i$ independently w.p. $p_i$ and is empty otherwise.

As probing reveals information, we differentiate between \emph{adaptive algorithms}, which behave differently conditioned on the realizations of probed items, and \emph{non-adaptive algorithms}.

\begin{definition}
    The \emph{adaptivity gap} is the ratio between the expectations obtained by the optimal adaptive algorithm and optimal non-adaptive algorithm.
\end{definition}

Asadpour and Nazerzadeh~\cite{AN15} give a tight result that the adaptivity gap for stochastic probing subject to a matroid constraint is $\frac{e}{e-1}$. The approach is to first define an auxiliary function $f'$, where $f'(S)$ is the expectation of $f$ upon probing items $\{X_i : i \in S\}$. It turns out that the multilinear extension $F'$ of $f'$ satisfies the property that $\max_{\mathbf{x} \in \calP_\calI} F'(\mathbf{x})$ is a $(1-\nicefrac{1}{e})$-approximation of the expectation obtained by the optimal adaptive algorithm.\footnote{We omit many of the finer details because our result does not alter this part of the analysis. Section 3 of \cite{AN15} covers this in depth.}

With this approximation result, the idea is to use $\mathbf{x}$ to design a non-adaptive algorithm. Simply probing each item w.p. $x_i$ may violate the matroid constraint, so Asadpour and Nazerzadeh design non-adaptive algorithms using pipage rounding. We go beyond matroids by designing non-adaptive algorithms using WNR distributions.

\StocProbing*

\begin{proof}
Our analysis follows that of \cite{AN15} until we obtain $\argmax_{\mathbf{x} \in \calP_\calI} F'(\mathbf{x})$. As $\calI$ admits WNR distributions, there exists a WNR distribution $\calD$ over $\calI$ with marginals $\mathbf{x}$. By \Cref{thm:CompIneqSufficient}, we have $\E_{S \sim \calD}[f'(S)] \geq F'(\mathbf{x})$. Therefore, the non-adaptive algorithm which samples $S \sim \calD$ and probes $\{X_i : i \in S\}$ obtains at least $F'(\mathbf{x})$ in expectation. No adaptive algorithm can obtain expectation greater than $\frac{e}{e-1} \cdot F'(\mathbf{x})$, so $\frac{e}{e-1}$ upper-bounds the adaptivity gap.
\end{proof}


\subsection{Contention Resolution Schemes} \label{sec:CRS}

\begin{definition}[CRS]
    A \emph{contention resolution scheme} (CRS) for a set system $\calI \subseteq 2^U$ and a distribution $\calD$ over $2^U$ with marginals $\mathbf{x}$ is a (possibly randomized) mapping $\pi_{\calI, \calD} : 2^U \rightarrow \calI$ such that for all $S \subseteq U$, we have $\pi_{\calI, \calD}(S) \subseteq S$.
\end{definition}

Contention resolution schemes have applications to submodular maximization as a rounding technique. The following is the simplest measure of performance for a CRS.

\begin{definition}[$c$-selectable CRS]
    For $c \in [0, 1]$, a set system $\calI \subseteq 2^U$, and a distribution $\calD$ over $2^U$ with marginals $\mathbf{x} \in \calP_\calI$, a CRS $\pi_{\calI, \calD}$ is \emph{$c$-selectable} if $\Pr_{S \sim \calD}[i \in \pi_{\calI, \calD}(S)] \geq c \cdot x_i$ for all $i \in U$.
\end{definition}

In submodular maximization, rounding fractional solutions is closely related to the multilinear extension, so study of CRSs is primarily centered around product distributions. However, as Dughmi~\cite{Dughmi-ICALP20,Dughmi-arXiv21} recently showed, CRSs over non-product distributions have applications in settings such as the Matroid Secretary Problem. We use Submodular Dominance to extend a selectability result to WNR distributions, which provides a direction by which other CRS results may be generalized to correlated distributions.

\CRS*

\begin{proof}
\cite{CVZ-SICOMP14} demonstrated this result for product distributions via strong LP duality. Following the same idea, Dughmi~\cite{Dughmi-ICALP20} reduced this result to proving Submodular Dominance:

\begin{proposition}[\cite{Dughmi-ICALP20}] \label{prop:CRSCondition}
    For a matroid $\calM$ and a distribution $\calD$ over $2^U$ with marginals $\mathbf{x} \in \calP_\calM$, there exists a $(1-\nicefrac{1}{e})$-selectable CRS if every submodular function $f : 2^U \rightarrow \mathbb{R}$ satisfies $\E_{S \sim \calD}[f(S)] \geq \E_{S \sim \mathbf{x}}[f(S)]$.
\end{proposition}
\Cref{thm:CRS} follows directly from \Cref{prop:CRSCondition} and \Cref{thm:CompIneqSufficient}.
\end{proof}

\section{Conclusion and Open Questions}

In this paper, we explore Submodular Dominance and its applications. In the process, we introduce a notion of negative dependence that we refer to as Weak Negative Regression (WNR), which is a natural generalization of both Negative Association (NA) and Negative Regression (NR) and may be of use in other applications. We prove that WNR distributions satisfy Submodular Dominance, and that all distributions satisfying Submodular Dominance also satisfy Negative Cylinder Dependence (NCD). Finally, we give a variety of applications for Submodular Dominance, improving the best known submodular prophet inequalities, developing new rounding techniques, and generalizing results for contention resolution schemes to negatively dependent distributions.

\medskip
\textbf{Sampling for More General Set Systems.}
Although our results for negatively distributions satisfying Submodular Dominance already have several applications, their usage could be broadened further by finding new techniques to generate negatively dependent distributions. An interesting future direction is to design algorithms to sample from negatively dependent distributions for more general set systems. For example, can we efficiently sample from a WNR/NA/NR distribution for any marginals in a given matroid polytope? We remark that \cite{PSV17} claimed such a result for NA distributions, but later, a gap in their proof was found. Max-entropy distributions over matroids are also not negatively dependent in general, as it is known that there exist matroids for which the uniform distribution (which is entropy-maximizing without constrained marginals) exhibits positive correlations~\cite{Jerrum06}.

\medskip
\textbf{Approximate Submodular Dominance.}
While we showed that NCD distributions do not always satisfy Submodular Dominance, one question is whether these weaker notions of negative dependence obtain constant-factor approximation variants of Submodular Dominance; that is, for a non-negative submodular function $f$, what distributions satisfy $\E_{S \sim \calD}[f(S)] \geq O(1) \cdot F(\mathbf{x})$? What about for monotone $f$? Another direction is to generalize Submodular Dominance to a larger class of functions. XOS functions are functions that can be expressed as the maximum of a collection of linear functions, and are a strict superset of submodular functions. While no non-product distribution satisfies ``XOS Dominance'' (consider $\max(X_1 + X_2, 1)$ and $\max(X_1, X_2)$, which are both XOS; the former decreases in expectation if $X_1$ and $X_2$ are negatively correlated, the latter if $X_1$ and $X_2$ are positively correlated), we might similarly ask if approximate versions hold for XOS functions.

\medskip
\textbf{Concentration Inequalities.}
Another important direction is understanding concentration inequalities for negatively dependent distributions. Submodular Dominance demonstrates that the expectation of negatively dependent distributions behaves favorably compared to product distributions, but we may also be interested in whether these distributions are concentrated around their mean. We know dimension-dependent concentration inequalities for arbitrary Lipschitz functions over NR distributions~\cite{GV18}. Proving dimension-independent concentration inequalities for submodular functions is an interesting future direction.

\appendix


\section{Missing Proofs}


\subsection{Weak Negative Regression Proofs} \label{sec:WNRProofs}

\WNRisSuperset*

\begin{proof}
We proved in \Cref{sec:WNR} that the union of NA and NR distributions is a subset of WNR, and that WNR distributions are a subset of NCD distributions. \Cref{sec:NotTightProofs} provides example distributions demonstrating strict containment of WNR in NCD.

To show strict containment of NA and NR in WNR, we consider the following distribution $\calD$, which was given by Joag-Dev and Proschan~\cite{JP83}.\footnote{They studied NA distributions over continuous random variables, and gave this distribution as an example that Negative Orthant Dependence (this is equivalent to NCD for Bernoulli random variables) does not imply NA.} We treat $\calD$ as a distribution over Bernoulli random variables $\mathbf{X} = (X_1, X_2, X_3, X_4)$ to simplify notation.

\begin{table}[ht]
\begin{center}
{\renewcommand{\arraystretch}{1.23}
\begin{tabular}{|C{1.3cm}|C{1.7cm}|C{1.5cm}|C{1.5cm}|C{1.5cm}|C{1.5cm}|C{1.5cm}|} \hline
    $\calD$ & \multicolumn{6}{c|}{$(X_1, X_2)$} \\ \hline \multirow{6}{*}{$(X_3, X_4)$}
    &          & (0, 0) & (0, 1) & (1, 0) & (1, 1) & Marginal \\ \cline{2-7}
    &   (0, 0) & 0.0577 & 0.0623 & 0.0623 & 0.0577 & 0.24     \\ \cline{2-7}
    &   (0, 1) & 0.0623 & 0.0677 & 0.0677 & 0.0623 & 0.26     \\ \cline{2-7}
    &   (1, 0) & 0.0623 & 0.0677 & 0.0677 & 0.0623 & 0.26     \\ \cline{2-7}
    &   (1, 1) & 0.0577 & 0.0623 & 0.0623 & 0.0577 & 0.24     \\ \cline{2-7}
    & Marginal & 0.24   & 0.26   & 0.26   & 0.24   &          \\ \hline
\end{tabular}}
\end{center}
\caption{A distribution which is WNR, but not NA or NR.}
\label{table:WNRDistribution}
\end{table}

$\calD$ violates NA because $\Cov_{\mathbf{X} \sim \calD}[X_1X_2, X_3X_4] > 0$, and $\calD$ violates NR because the conditional expectation $\E_{\mathbf{X} \sim \calD}[X_1X_2 \given X_3 = 1, X_4]$ is increasing in $X_4$.

By observing that the value in column 2 is larger than in column 4 for any row of \Cref{table:WNRDistribution}, we see that for any $x_3, x_4 \in \{0, 1\}$, the condtional expectation $\E_{\mathbf{X} \sim \calD}[X_2 \given X_3 = x_3,\, X_4 = x_4,\, X_1]$ is decreasing in $X_1$. Therefore, we can convert the distribution $\calD$ conditioned on $X_1 = 0$ into the distribution $\calD$ conditioned on $X_1 = 1$ by only transferring probability mass ``downwards," which cannot increase the expectation of a monotone function $f : \{0, 1\}^4 \rightarrow \mathbb{R}$. Thus, for any such function which does not depend on $X_1$,
\[
    \E_{\mathbf{X} \sim \calD}[f(\mathbf{X}) \given X_1 = 1] \quad \leq \quad \E_{\mathbf{X} \sim \calD}[f(\mathbf{X}) \given X_1 = 0] \enspace ,
\]
which is the WNR condition. Since $X_1, X_2$ and $(X_1, X_2), (X_3, X_4)$ are both exchangeable, we can repeat this analysis for all $X_i$. Thus, $\calD$ is WNR, but neither NA nor NR.
\end{proof}

\WNRProperties*

\begin{proof}
Closure under projection follows trivially because the WNR condition \eqref{eq:WNRCondition} is satisfied for all monotone functions, and monotone functions restricted to a subset of elements are still monotone.

For closure under products, let $\calA$ and $\calB$ be WNR distributions over $2^A$ and $2^B$ for disjoint $A, B$, and let $\calD$ be their product. WLOG, fix $i \in A$ and a monotone function $f : 2^{A \cup B} \rightarrow \mathbb{R}$. Since $\calA$ and $\calB$ are independent, $\E_{S \sim \calD}[f(S \setminus i) \given S \cap B = T] = \E_{S \sim \calA}[f((S \setminus i) \cup T)]$. Since $\calA$ is WNR and $f(S \cup T)$ is still a monotone function,
\[
    \E_{S \sim \calD}[f(S \setminus i) \given S \cap B = T,\, i \in S] \quad \leq \quad \E_{S \sim \calD}[f(S \setminus i) \given S \cap B = T,\, i \not\in S] \enspace .
\]
Taking expectations over $S \cap B$ gives $\E_{S \sim \calD}[f(S \setminus i) \given i \in S] \leq \E_{S \sim \calD}[f(S \setminus i) \given i \not\in S]$, completing the proof.
\end{proof}


\subsection{Submodular Dominance Example Distributions} \label{sec:NotTightProofs}

\CompIneqUnnecessary*

\begin{proof}
Using the definition of $\calD_k$ from the proof of \Cref{thm:CompIneqSufficient}, notice that $\calD_1$ is a product distribution. Therefore, we only need the analysis in \Cref{claim:CompIneqStep2} to follow, which only requires that the WNR condition \eqref{eq:WNRCondition} holds for $i = 1$. $\calD$ conditioned on $1 \in S$ samples $\emptyset$, $\{2\}$, and $\{3\}$ w.p. $\nicefrac{1}{3}$, and $\calD$ conditioned on $1 \not\in S$ samples $\emptyset$, $\{2\}$, and $\{2, 3\}$ w.p. $\nicefrac{1}{3}$, which cannot obtain lower expectation for a monotone function, so the WNR condition holds for $i = 1$ and $\calD$ satisfies Submodular Dominance.

$\calD$ violates WNR for $f$ and $i = 3$ because $\calD$ conditioned on $3 \in S$ always samples either $1$ or $2$, whereas $\calD$ conditioned on $3 \not\in S$ can sample $\emptyset$. Thus, there exist distributions which satisfy Submodular Dominance but violate WNR.
\end{proof}

\CompIneqInsufficient*

\begin{proof}
Notice that $\calD$ is identical under permutations of elements. Further, because $i$ or $[4] \setminus i$ is returned with equal probability, we have the property that for any $T \subseteq [4]$,
\[
    \Pr_{S \sim \calD}[T \subseteq S] \quad = \quad \Pr_{S \sim \calD}[T \subseteq S^c] \enspace .
\]
Therefore, it is sufficient to show that $\Pr_{S \sim \calD}[1, 2 \in S] \leq \nicefrac{1}{4}$, $\Pr_{S \sim \calD}[1, 2, 3 \in S] \leq \nicefrac{1}{8}$, and $\Pr_{S \sim \calD}[1, 2, 3, 4 \in S] \leq \nicefrac{1}{16}$ to prove $\calD$ is NCD.
\begin{itemize}[topsep=0pt,itemsep=0pt]
    \item For $1, 2 \in S$, we need to choose $i = 3, 4$, then return $[4] \setminus i$. This occurs w.p. $\nicefrac{1}{2} \cdot \nicefrac{1}{2} = \nicefrac{1}{4}$.
    
    \item For $1, 2, 3 \in S$, we need to choose $i = 4$, then return $[4] \setminus i$. This occurs w.p. $\nicefrac{1}{4} \cdot \nicefrac{1}{2} = \nicefrac{1}{8}$.
    
    \item There is no way for $1, 2, 3, 4 \in S$.
\end{itemize}
Thus, $\calD$ is NCD. $f$ is a matroid rank function, so it is submodular. Letting $\mathbf{x}$ be the marginals of $\calD$, a simple expected value computation shows that $\E_{S \sim \calD}[f(S)] = \nicefrac{12}{8} < \nicefrac{13}{8} = F(\mathbf{x})$, so $\calD$ violates Submodular Dominance.
\end{proof}

\CompIneqInsufficientH*

\begin{proof}
This example extends the previous example to a homogeneous distribution. Similar to the previous example, $\calD$ is closed under permutations of $A$, permutations of $B$, and swaps of $A$ with $B$. Since we already showed any $T \subseteq A$ satisfies the NCD condition, it is sufficient to show that $\Pr_{S \sim \calD}[1, 5 \in S] \leq \nicefrac{1}{4}$, $\Pr_{S \sim \calD}[1, 2, 5 \in S] \leq \nicefrac{1}{8}$, $\Pr_{S \sim \calD}[1, 2, 5, 6 \in S] \leq \nicefrac{1}{16}$, and $\Pr_{S \sim \calD}[1, 2, 3, 5 \in S] \leq \nicefrac{1}{16}$ (since only sets of size $4$ are drawn, if $\vert T \vert > 4$ it automatically satisfies the NCD condition).
\begin{itemize}[topsep=0pt,itemsep=0pt]
    \item For $1, 5 \in S$, we need to choose $i = 1$ and $j = 6, 7, 8$, then return $i \cup (B \setminus j)$, or choose $i = 2, 3, 4$ and $j = 5$, then return $(A \setminus i) \cup j$. This occurs w.p. $2 \cdot \nicefrac{1}{4} \cdot \nicefrac{3}{4} \cdot \nicefrac{1}{2} = \nicefrac{3}{16} \leq \nicefrac{1}{4}$.
    
    \item For $1, 2, 5 \in S$, we need to choose $i = 3, 4$ and $j = 5$, then return $(A \setminus i) \cup j$. This occurs w.p. $\nicefrac{1}{2} \cdot \nicefrac{1}{4} \cdot \nicefrac{1}{2} = \nicefrac{1}{16} \leq \nicefrac{1}{8}$.
    
    \item There is no way for $1, 2, 5, 6 \in S$.
    
    \item For $1, 2, 3, 5 \in S$, we need to choose $i = 4$ and $j = 5$, then return $(A \setminus i) \cup j$. This occurs w.p. $\nicefrac{1}{4} \cdot \nicefrac{1}{4} \cdot \nicefrac{1}{2} = \nicefrac{1}{32} \leq \nicefrac{1}{16}$.
\end{itemize}
Thus, $\calD$ is NCD, and we again have $\E_{S \sim \calD}[f(S)] = \nicefrac{12}{8} < \nicefrac{13}{8} = \E_{S \sim \mathbf{x}}[f(S)]$, so $\calD$ violates Submodular Dominance.
\end{proof}


\subsection{Product of Singletons is a Convex Combination of Product Distributions} \label{sec:POSasProduct}

\POSasProduct*

\begin{proof}
For some weights $p_S$, we can rewrite RHS of \eqref{eq:POSasProduct} as
\[
    \sum_{S \subseteq E} g(S) \cdot p_S \enspace .
\]
Our approach is to show that $p_T = \Pr_{S \sim \calD}[S = T]$ for any $T \subseteq E$. Then the summation is simply the expectation of $g$ over $\calD$ and we are finished.

Fix some set $T \subseteq E$. The distributions for which $T$ is in the image of $\calD_\mathbf{u}$ are those where for all $ij \in T$, $u_i = j$. Therefore,
\[
    p_T \quad = \quad \sum_{\substack{\mathbf{u} \in [m]^n \\ u_i = j \, \forall ij \in T}} \bigg(\Pr_{S \sim \calD_\mathbf{u}}\big[S = T\big] \cdot \prod_{i \in U} \frac{x_{iu_i}}{x_i}\bigg) \enspace .
\]
Let $T^* \subseteq U$ be a set where $i \in T^*$ if there exists some $j$ for which $ij \in T$. Then,
\[
    p_T \quad = \quad \sum_{\substack{\mathbf{u} \in [m]^n \\ u_i = j \, \forall ij \in T}} \bigg(\prod_{i \in T^*} x_i \prod_{i \not\in T^*} (1-x_i) \cdot \prod_{i \in T^*} \frac{x_{iu_i}}{x_i} \prod_{i \not\in T^*} \frac{x_{iu_i}}{x_i}\bigg) \enspace .
\]
We combine the products over $i \in T^*$, and move the first product over $i \not\in T^*$ outside the summation as the inner term does not depend on $\mathbf{u}$.
\[
    p_T \quad = \quad \prod_{i \not\in T^*} (1-x_i) \cdot \sum_{\substack{\mathbf{u} \in [m]^n \\ u_i = j \, \forall ij \in T}} \bigg(\prod_{i \in T^*} x_{iu_i} \prod_{i \not\in T^*} \frac{x_{iu_i}}{x_i}\bigg) \enspace .
\]
Because the summation is restricted to $\mathbf{u}$ where $u_i = j$ for all $ij \in T$, the coordinates of $\mathbf{u}$ for $i \in T^*$ can only take one value. Thus, the product over $i \in T^*$ is always the same, and can be factored out of the summation.
\[
    p_T \quad = \quad \prod_{i \not\in T^*} (1-x_i) \cdot \prod_{ij \in T} x_{ij} \cdot \sum_{\substack{\mathbf{u} \in [m]^n \\ u_i = j \, \forall ij \in T}} \bigg(\prod_{i \not\in T^*} \frac{x_{iu_i}}{x_i}\bigg) \enspace .
\]
As we just observed, the summation only enforces a condition on $i \in T^*$, so we sum up over all possible $u_i \in [m]$ for $i \not\in T^*$. We can rewrite this as
\begin{align*}
    p_T \quad &= \quad \prod_{i \not\in T^*} (1-x_i) \cdot \prod_{ij \in T} x_{ij} \cdot \prod_{i \not\in T^*} \bigg(\sum_{j \in [m]} \frac{x_{ij}}{x_i}\bigg) \\
    &= \quad \prod_{i \not\in T^*} (1-x_i) \cdot \prod_{ij \in T} x_{ij} \cdot \prod_{i \not\in T^*} \bigg(\frac{1}{x_i} \cdot \sum_{j \in [m]} x_{ij}\bigg) \\
    &= \quad \prod_{i \not\in T^*} (1-x_i) \cdot \prod_{ij \in T} x_{ij} \cdot \prod_{i \not\in T^*} \bigg(\frac{1}{x_i} \cdot x_i\bigg) \\
    &= \quad \prod_{i \not\in T^*} (1-x_i) \cdot \prod_{ij \in T} x_{ij} \\
    &= \quad \Pr_{S \sim \calD}[S = T] \enspace .
\end{align*}
Since these computations follow for any $T \subseteq E$, we have
\[
    \E_{S \sim \calD}[g(S)] \quad = \quad \sum_{S \subseteq E} g(S) \cdot p_S \quad = \quad \sum_{\mathbf{u} \in [m]^n} \bigg(\E_{S \sim \calD_\mathbf{u}}[g(S)] \cdot \prod_{i \in U} \frac{x_{iu_i}}{x_i}\bigg) \enspace ,
\]
which completes the proof.
\end{proof}

\begin{small}
\bibliographystyle{alpha}
\bibliography{bibliography}

\begin{thebibliography}{NRTV07}

\bibitem[Ala11]{Alaei-FOCS11}
Saeed Alaei.
\newblock Bayesian combinatorial auctions: Expanding single buyer mechanisms to
  many buyers.
\newblock In {\em {IEEE} 52nd Annual Symposium on Foundations of Computer
  Science, {FOCS} 2011, Palm Springs, CA, USA, October 22-25, 2011}, pages
  512--521, 2011.

\bibitem[AN16]{AN15}
Arash Asadpour and Hamid Nazerzadeh.
\newblock Maximizing stochastic monotone submodular functions.
\newblock {\em Manag. Sci.}, 62(8):2374--2391, 2016.

\bibitem[Bac13]{Bach-Book13}
Francis~R. Bach.
\newblock Learning with submodular functions: {A} convex optimization
  perspective.
\newblock {\em Found. Trends Mach. Learn.}, 6(2-3):145--373, 2013.

\bibitem[BKNS12]{BKNS12}
Nikhil Bansal, Nitish Korula, Viswanath Nagarajan, and Aravind Srinivasan.
\newblock Solving packing integer programs via randomized rounding with
  alterations.
\newblock {\em Theory Comput.}, 8(1):533--565, 2012.

\bibitem[BSZ19]{BSZ-Random19}
Domagoj Bradac, Sahil Singla, and Goran Zuzic.
\newblock {(Near)} optimal adaptivity gaps for stochastic multi-value probing.
\newblock In {\em Approximation, Randomization, and Combinatorial Optimization.
  Algorithms and Techniques, {APPROX/RANDOM} 2019, September 20-22, 2019,
  Massachusetts Institute of Technology, Cambridge, MA, {USA}}, pages
  49:1--49:21, 2019.

\bibitem[CCPV07]{CCPV07}
Gruia C{\u{a}}linescu, Chandra Chekuri, Martin P{\'{a}}l, and Jan
  Vondr{\'{a}}k.
\newblock Maximizing a submodular set function subject to a matroid constraint
  (extended abstract).
\newblock In {\em Integer Programming and Combinatorial Optimization, 12th
  International {IPCO} Conference, Ithaca, NY, USA, June 25-27, 2007,
  Proceedings}, pages 182--196, 2007.

\bibitem[CCPV11]{CCPV-SICOMP11}
Gruia Calinescu, Chandra Chekuri, Martin Pal, and Jan Vondr{\'a}k.
\newblock Maximizing a monotone submodular function subject to a matroid
  constraint.
\newblock {\em SIAM Journal on Computing}, 40(6):1740--1766, 2011.

\bibitem[CHMS10]{CHMS-STOC10}
Shuchi Chawla, Jason~D. Hartline, David~L. Malec, and Balasubramanian Sivan.
\newblock Multi-parameter mechanism design and sequential posted pricing.
\newblock In {\em Proceedings of the 42nd {ACM} Symposium on Theory of
  Computing, {STOC} 2010, Cambridge, Massachusetts, USA, 5-8 June 2010}, pages
  311--320, 2010.

\bibitem[CL21]{CL21}
Chandra Chekuri and Vasilis Livanos.
\newblock On submodular prophet inequalities and correlation gap.
\newblock In Ioannis Caragiannis and Kristoffer~Arnsfelt Hansen, editors, {\em
  Algorithmic Game Theory - 14th International Symposium, {SAGT} 2021, Aarhus,
  Denmark, September 21-24, 2021, Proceedings}, volume 12885 of {\em Lecture
  Notes in Computer Science}, page 410. Springer, 2021.

\bibitem[CV04]{CV04}
Tasos Christofides and Eutichia Vaggelatou.
\newblock A connection between supermodular ordering and positive/negative
  association.
\newblock {\em Journal of Multivariate Analysis}, 88:138--151, 01 2004.

\bibitem[CVZ09]{CVZ-arXiv09}
Chandra Chekuri, Jan Vondr{\'a}k, and Rico Zenklusen.
\newblock Dependent randomized rounding for matroid polytopes and applications.
\newblock {\em arXiv preprint arXiv:0909.4348}, 2009.

\bibitem[CVZ10]{CVZ-FOCS10}
Chandra Chekuri, Jan Vondr{\'{a}}k, and Rico Zenklusen.
\newblock Dependent randomized rounding via exchange properties of
  combinatorial structures.
\newblock In {\em 51th Annual {IEEE} Symposium on Foundations of Computer
  Science, {FOCS} 2010, October 23-26, 2010, Las Vegas, Nevada, {USA}}, pages
  575--584, 2010.

\bibitem[CVZ14]{CVZ-SICOMP14}
Chandra Chekuri, Jan Vondr{\'{a}}k, and Rico Zenklusen.
\newblock Submodular function maximization via the multilinear relaxation and
  contention resolution schemes.
\newblock {\em {SIAM} J. Comput.}, 43(6):1831--1879, 2014.

\bibitem[Dug20]{Dughmi-ICALP20}
Shaddin Dughmi.
\newblock The outer limits of contention resolution on matroids and connections
  to the secretary problem.
\newblock In {\em 47th International Colloquium on Automata, Languages, and
  Programming, {ICALP} 2020, July 8-11, 2020, Saarbr{\"{u}}cken, Germany
  (Virtual Conference)}, pages 42:1--42:18, 2020.

\bibitem[Dug22]{Dughmi-arXiv21}
Shaddin Dughmi.
\newblock Matroid secretary is equivalent to contention resolution.
\newblock In {\em 13th Innovations in Theoretical Computer Science Conference,
  {ITCS} 2022, January 31 - February 3, 2022, Berkeley, CA, {USA}}, pages
  58:1--58:23, 2022.

\bibitem[EKM21]{EKM-COLT21}
Hossein Esfandiari, Amin Karbasi, and Vahab~S. Mirrokni.
\newblock Adaptivity in adaptive submodularity.
\newblock In Mikhail Belkin and Samory Kpotufe, editors, {\em Conference on
  Learning Theory, {COLT} 2021, 15-19 August 2021, Boulder, Colorado, {USA}},
  volume 134 of {\em Proceedings of Machine Learning Research}, pages
  1823--1846. {PMLR}, 2021.

\bibitem[EN16]{EneN-FOCS16}
Alina Ene and Huy~L. Nguyen.
\newblock Constrained submodular maximization: Beyond 1/e.
\newblock In {\em {IEEE} 57th Annual Symposium on Foundations of Computer
  Science, {FOCS} 2016, 9-11 October 2016, Hyatt Regency, New Brunswick, New
  Jersey, {USA}}, pages 248--257, 2016.

\bibitem[FNS11]{FNS-FOCS11}
Moran Feldman, Joseph Naor, and Roy Schwartz.
\newblock A unified continuous greedy algorithm for submodular maximization.
\newblock In {\em {IEEE} 52nd Annual Symposium on Foundations of Computer
  Science, {FOCS} 2011, Palm Springs, CA, USA, October 22-25, 2011}, pages
  570--579, 2011.

\bibitem[FSZ16]{FSZ15}
Moran Feldman, Ola Svensson, and Rico Zenklusen.
\newblock Online contention resolution schemes.
\newblock In Robert Krauthgamer, editor, {\em Proceedings of the Twenty-Seventh
  Annual {ACM-SIAM} Symposium on Discrete Algorithms, {SODA} 2016, Arlington,
  VA, USA, January 10-12, 2016}, pages 1014--1033. {SIAM}, 2016.

\bibitem[Fuj05]{Fujishige-Book05}
Satoru Fujishige.
\newblock {\em Submodular functions and optimization}.
\newblock Elsevier, 2005.

\bibitem[GN13]{GN-IPCO13}
Anupam Gupta and Viswanath Nagarajan.
\newblock A stochastic probing problem with applications.
\newblock In {\em Integer Programming and Combinatorial Optimization - 16th
  International Conference, {IPCO} 2013, Valpara{\'{\i}}so, Chile, March 18-20,
  2013. Proceedings}, pages 205--216, 2013.

\bibitem[GNS16]{GNS-SODA16}
Anupam Gupta, Viswanath Nagarajan, and Sahil Singla.
\newblock Algorithms and adaptivity gaps for stochastic probing.
\newblock In {\em Twenty-Seventh Annual {ACM-SIAM} Symposium on Discrete
  Algorithms, {SODA} 2016, Arlington, VA, USA, January 10-12, 2016}, pages
  1731--1747, 2016.

\bibitem[GV18]{GV18}
Kevin Garbe and Jan Vondrak.
\newblock Concentration of lipschitz functions of negatively dependent
  variables.
\newblock {\em arXiv preprint arXiv:1804.10084}, 2018.

\bibitem[HKS07]{HKS-AAAI07}
Mohammad~Taghi Hajiaghayi, Robert~D. Kleinberg, and Tuomas Sandholm.
\newblock Automated online mechanism design and prophet inequalities.
\newblock In {\em Proceedings of the Twenty-Second {AAAI} Conference on
  Artificial Intelligence, July 22-26, 2007, Vancouver, British Columbia,
  Canada}, pages 58--65, 2007.

\bibitem[JDP83]{JP83}
Kumar Joag-Dev and Frank Proschan.
\newblock Negative association of random variables with applications.
\newblock {\em The Annals of Statistics}, 11(1):286--295, 1983.

\bibitem[Jer06]{Jerrum06}
Mark Jerrum.
\newblock Two remarks concerning balanced matroids.
\newblock {\em Comb.}, 26(6):733--742, 2006.

\bibitem[KS77]{Krengel-Journal77}
Ulrich Krengel and Louis Sucheston.
\newblock Semiamarts and finite values.
\newblock {\em Bull. Am. Math. Soc}, 1977.

\bibitem[KS78]{Krengel-Journal78}
Ulrich Krengel and Louis Sucheston.
\newblock On semiamarts, amarts, and processes with finite value.
\newblock {\em Advances in Prob}, 4:197--266, 1978.

\bibitem[KW12]{KW-STOC12}
Robert Kleinberg and S.~Matthew Weinberg.
\newblock Matroid prophet inequalities.
\newblock In {\em Proceedings of the 44th Symposium on Theory of Computing
  Conference, {STOC} 2012, New York, NY, USA, May 19 - 22, 2012}, pages
  123--136, 2012.

\bibitem[NRTV07]{NRTV2007}
Noam Nisan, Tim Roughgarden, {\'{E}}va Tardos, and Vijay~V. Vazirani, editors.
\newblock {\em Algorithmic Game Theory}.
\newblock Cambridge University Press, 2007.

\bibitem[NW78]{nemhauserWolsey78}
George~L Nemhauser and Laurence~A Wolsey.
\newblock Best algorithms for approximating the maximum of a submodular set
  function.
\newblock {\em Mathematics of operations research}, 3(3):177--188, 1978.

\bibitem[Pem00]{Pemantle99}
Robin Pemantle.
\newblock Towards a theory of negative dependence.
\newblock {\em Journal of Mathematical Physics}, 41(3):1371--1390, 2000.

\bibitem[PSV17]{PSV17}
Yuval Peres, Mohit Singh, and Nisheeth~K. Vishnoi.
\newblock Random walks in polytopes and negative dependence.
\newblock In {\em 8th Innovations in Theoretical Computer Science Conference,
  {ITCS} 2017, January 9-11, 2017, Berkeley, CA, {USA}}, pages 50:1--50:10,
  2017.

\bibitem[RS17]{RS-SODA17}
Aviad Rubinstein and Sahil Singla.
\newblock Combinatorial prophet inequalities.
\newblock In {\em Proceedings of the Twenty-Eighth Annual {ACM-SIAM} Symposium
  on Discrete Algorithms, {SODA} 2017, Barcelona, Spain, Hotel Porta Fira,
  January 16-19}, pages 1671--1687, 2017.

\bibitem[Rub16]{Rubinstein-STOC16}
Aviad Rubinstein.
\newblock Beyond matroids: secretary problem and prophet inequality with
  general constraints.
\newblock In {\em Proceedings of the 48th Annual {ACM} {SIGACT} Symposium on
  Theory of Computing, {STOC} 2016, Cambridge, MA, USA, June 18-21, 2016},
  pages 324--332, 2016.

\bibitem[Sch03]{Schrijver-Book03}
Alexander Schrijver.
\newblock {\em Combinatorial optimization: polyhedra and efficiency},
  volume~24.
\newblock Springer Science \& Business Media, 2003.

\bibitem[Sha00]{Shao00}
Qi-Man Shao.
\newblock A comparison theorem on moment inequalities between negatively
  associated and independent random variables.
\newblock {\em Journal of Theoretical Probability}, 13(2):343--356, 2000.

\bibitem[Von07]{Vondrak07}
Jan Vondr{\'a}k.
\newblock {\em Submodularity in combinatorial optimization}.
\newblock PhD thesis, Univerzita Karlova, Matematicko-fyzik{\'a}ln{\'\i}
  fakulta, 2007.

\bibitem[Von08]{Vondrak-STOC08}
Jan Vondr{\'a}k.
\newblock Optimal approximation for the submodular welfare problem in the value
  oracle model.
\newblock In {\em Proceedings of the fortieth annual ACM symposium on Theory of
  computing}, pages 67--74, 2008.

\end{thebibliography}
\end{small}

\end{document}